\documentclass[twocolumn,superscriptaddress,showpacs,nofootinbib,hyperref,notitlepage]{revtex4-1}
\usepackage{amsmath,verbatim,latexsym,amssymb,graphicx,subcaption,indentfirst,mathrsfs,mathtools,amsthm,bbm,bm,enumitem}
\usepackage[font=small,labelfont=bf,format=plain,justification=raggedright,singlelinecheck=false]{caption}

\newtheorem{theorem}{Theorem}
\newtheorem{lemma}{Lemma}

\newtheorem{definition}{Definition}

\newtheorem{conjecture}{Conjecture}

\makeatletter
\newcommand\footnoteref[1]{\protected@xdef\@thefnmark{\ref{#1}}\@footnotemark}
\makeatother

\usepackage{soul}

\usepackage{color}
\usepackage[dvipsnames]{xcolor}
\usepackage[normalem]{ulem}

\DeclareMathOperator{\supp}{supp}

\newcommand{\Tr}{{\rm Tr}}
\newcommand{\I}{ {(i)} }

\renewcommand\th{ {\rm th} }

\newcommand*{\ket}[1]{\lvert #1 \rangle}

\newcommand*{\ketbra}[2]{\lvert #1 \rangle\!\langle #2 \rvert}

\def\soc{{\rm soc}}
\newcommand{\class}{{\rm cl}}
\newcommand{\Dens}{\mathcal{D}}
\newcommand{\BH}{ \mathcal{B} }
\newcommand{\Hil}{ \mathcal{H} }

\newcommand{\C}{ \mathcal{C} }  % Classical constitution
\newcommand{\E}{ \mathcal{E} }  % Quantum constitution
\newcommand{\EQMR}{\mathcal{E}_\QMR}
\newcommand{\EQMRTwo}{\mathcal{E}_\QMR^{(2)}}
\newcommand{\EQMRTh}{\mathcal{E}_\QMR^{(3)}} 
\newcommand{\QMR}{{\rm QMR}}

\begin{document}

\title{Quantum voting and violation of Arrow's Impossibility Theorem}

%
% Authors
%
\author{Ning~Bao\footnote{E-mail: ningbao@its.caltech.edu}}
\affiliation{Institute for Quantum Information and Matter, California Institute of Technology, Pasadena, CA 91125, USA}
\affiliation{Walter Burke Institute for Theoretical Physics, California Institute of Technology, Pasadena, CA 91125}

\author{Nicole~Yunger~Halpern\footnote{E-mail: nicoleyh@caltech.edu}}
\affiliation{Institute for Quantum Information and Matter, California Institute of Technology, Pasadena, CA 91125, USA}

\nopagebreak

\date{\today}

% PACS: http://www.aip.org/publishing/pacs/pacs-alphabetical-index
%             http://www.aip.org/publishing/pacs/pacs-reg00#05
\pacs{03.67.Ac}   % Quantum information

\begin{abstract}
We propose a quantum voting system, in the spirit of quantum games 
such as the quantum Prisoner's Dilemma. 
Our scheme enables a constitution
to violate a quantum analog of Arrow's Impossibility Theorem.
Arrow's Theorem is a claim proved deductively in economics: 
Every (classical) constitution 
endowed with three innocuous-seeming properties is a dictatorship.
We construct quantum analogs of
constitutions, of the properties, and of Arrow's Theorem.
A quantum version of majority rule, we show,
violates this Quantum Arrow Conjecture.
Our voting system allows for tactical-voting strategies
reliant on entanglement, interference, and superpositions.
This contribution to quantum game theory helps elucidate how 
quantum phenomena can be harnessed for strategic advantage.
\end{abstract}
\maketitle{}

Today's voting systems are classical.
Societies hold elections to determine
which of several candidates will win an office.
Each voter ranks the candidates, forming a \emph{preference}.
Voters' preferences are combined deterministically
according to some rule set, or constitution.
What if citizens could entangle, superpose, and interfere preferences?
We formulate a quantum voting system,
in the tradition of quantum games,
that highlights the power of quantum resources.

Quantum game theory has flourished 
over the past several years~\cite{FlitneyA02,Brunner_13_Connection}.
In a classical game, players can perform
only local operations and classical communications.
Each player can prepare and measure only systems in his/her lab.
Players can communicate only via classical channels
(e.g., by telephone), if at all.

Examples include the Prisoner's Dilemma.
Suppose that the police arrest two suspected criminals.
The suspects are isolated in separate cells.
If neither suspect confesses, 
each will receive a lenient sentence
(e.g., one year in jail).
If both suspects confess, both will receive moderate sentences
(e.g., two years).
If just one suspect confesses, s/he will receive no sentence.
The other suspect will suffer a heavy penalty (e.g., three years).
Unable to communicate with the other prisoner, 
each suspect can optimize his/her future by confessing.
Both suspects would benefit more
if they could agree to remain silent.
The Prisoner's Dilemma consists of
the tension between 
(i) the optimal strategy attainable
and (ii) the optimal strategy that the prisoners could attain
if they could communicate.

Quantizing the game diminishes the tension~\cite{EisertWL99}.
Eisert \emph{et al.} associate each prisoner with a Hilbert space.
They translate each prisoner's options 
(to cooperate with the police and to defect) 
into basis elements ($\ket{C}$ and $\ket{D}$).
The game becomes a quantum circuit.
Measuring the prisoners' joint state
determines their penalties.
This quantization alters the landscape of 
possible outcomes and strategies.

Similar insights result from quantizing
the penny-flipping game~\cite{Meyer99}, the Monty Hall problem~\cite{Flitney_02_Monty,Dariano_02_Monty},
and Conrad's Game of Life~\cite{Bleh_12_QGoL,Arrighi_10_QGoL}.
% Cite the power that I sent to the QGoL group in fall 2016?
A game elucidates 
the canonical demonstration of entanglement's power:
Clauser, Holt, Shimony, and Hauser (CHSH) reformulated Bell's Theorem
in terms of a protocol cast as 
``the CHSH game''~\cite{CHSH_69_Proposed,Wilde_12_Quantum}.

Elections have been cast in game-theoretic terms~\cite{Ordeshook_03_Game}.
Elections therefore merit generalization with quantum theory.
Upshots of quantization, we show,
include a violation of a quantum analog of Arrow's Impossibility Theorem,
as well as quantum voting strategies.

Arrow's Theorem is a result derived, in economics,
from deductive logic and definitions~\cite{Arrow50}.
According to the theorem, 
every constitution that has 
three innocuous-seeming properties  
(transitivity, unanimity, and independence of irrelevant alternatives,
defined below)
is a dictatorship. 
Arrow's Theorem is surprisingly deep 
and has fundamentally impacted game theory and voting theory
(e.g.,~\cite{Arrow_02_Handbook}).
Yet Arrow's Theorem derives from classical logic.
A quantum version, we find, is false.

Classical constitutions disobey Arrow's Theorem
under certain conditions.
For example, Black supplements Arrow's three postulates 
with extra assumptions~\cite{Black69}.
The extra assumptions, he argues, are properties
as reasonable as Arrow's
for a constitution to have.
No constitution, he shows, can satisfy
Arrow's postulates and the extras
while being a nondictatorship. 
Some sets of votes, however,
prevent constitutions that satisfy Black's extra assumptions
from satisfying all of Arrow's postulates.
Probabilistic mixtures of votes, too, evade Arrow's Theorem.
Suppose that a voter can pledge 
40\% of his/her support to candidate Alice,
40\% to Bob, and 20\% to Charlie.
Such voters can form a constitution
not subject to Arrow's Theorem.
Yet such constitutions do not satisfy all of Arrow's postulates~\cite{Riker_82_Liberalism,Kalai_02_Fourier,Mossel09}. 
We cleave to Arrow's postulates
and avoid restricting voters' preferences.
Rather, we recast Arrow's scheme in quantum terms. 
% Ultrafilters and lattice projections of von Neumann algebras have been marshaled against Arrow's Theorem~\cite{Segre08}. Our strategy is less technical.

Like Black's extra postulates and like probabilistic votes,
alternative classical voting schemes evade Arrow's Theorem.
Engaging in \emph{range voting}, a voter assigns each candidate 
a number of points
independently of the other candidates~\cite{Harsanyi_55_Cardinal,Harsanyi_77_Rational}.
Whichever candidate receives the most points wins.
Voters behave identically when using \emph{majority judgment}~\cite{Balinski_10_Majority}.
The majority-judgment winner has the highest median number of points.
Under \emph{approval voting}, each voter assigns each candidate
a thumbs-up or a thumbs-down~\cite{Ottewell_87_Arithmetic}.
Range voting, majority judgment, and approval voting 
contrast with \emph{ordinal voting}.
Ordinal-voting citizens rank candidates.
Arrow's Theorem governs just ordinal voting.
Yet a generalization of Arrow's Theorem, the Gibbard-Satterthwaite (GS) Theorem~\cite{Benoit_00_Gibbard,Reny_01_Arrow's},
governs majority judgment and approval voting.
(Range voting is not a scheme of the class
governed by the GS Theorem,
just as range voting is not a scheme of the class
governed by Arrow's Theorem.)
% Range voting disobeys the GS theorem. http://rangevoting.org/GibbSat.html
Like the GS Theorem, our Quantum Arrow Conjecture 
generalizes Arrow's Theorem.
Yet classicality does not constrain our generalization
as it constrains the GS Theorem.
Like range voting, majority judgment, and approval voting,
our voting scheme is not precisely ordinal. 
Yet ordinal rankings form the basis for our quantum votes' Hilbert spaces,
as discussed below.

In addition to disproving a Quantum Arrow Conjecture,
we present four quantum strategic-voting tactics.
How one should vote is not always clear,
even to opinionated citizens.
You might favor a candidate unlikely to win, for example.
Voting for a more likely candidate 
whose policies you could tolerate
can optimize the election's outcome.
Strategic voting is 
the submission of a preference 
other than one's opinion,
in a competition amongst three or more candidates~\cite{Barbera_01_Intro}.
Quantizing voting unlocks new voting strategies.
We exhibit three tactics reliant on entanglement
and one reliant on interference and superpositions.

Earlier work on quantum voting has focused on 
privacy, security, and cryptography~\cite{VaccaroASC07,HilleryZBB06,JiangHNXZ12}.
These references answer questions such as 
``How can voters and election officials hinder cheaters?''
In contrast, we draw inspiration from game theory.

The paper is organized as follows.
First, we introduce our quantum voting system.
We define quantum analogs of properties of classical constitutions.
Four of these properties appear in Arrow's Theorem,
which we review and quantize.
We disprove the conjecture by a counterexample.
The counterexample relies on the quantization of 
a fifth property available to classical constitutions: majority rule.
Finally, we present three strategic-voting strategies
based on entanglement
and one strategy based on interference.

\section{Quantum voting system}
\label{section:InitialDefns}

A voting system involves a \emph{society} 
that consists of \emph{voters}. 
\emph{Candidates} $a, b, \ldots, m$ vie for office.
Each voter ranks the candidates, forming a \emph{preference}.
A preference is a transitive ordered list.
Each candidate is ranked above, ranked below, or tied with 
each other candidate: $a > b$, $a < b$, or $a = b$. 
% Preferences can be nonstrict, as in $a \geq b$.
A list is transitive if $a \geq b$ and $b \geq c$, 
together, imply $a \geq c$. 

The voters' preferences form a \emph{profile}.
The profile serves as input to a \emph{constitution}
during an \emph{election}.
We focus on elections that feature at least three candidates.
The constitution combines the voters' preferences,
forming \emph{society's preference}.
Society's preference implies which candidate wins.

We quantize this classical election scheme.
Our strategy resembles that of Eisert \emph{et al.}~\cite{EisertWL99}.
Their quantum game consists of 
a general quantum process:
a preparation procedure, an evolution, and a measurement~\cite{NielsenC10}.
So does our quantum voting scheme.
We introduce a Hilbert-space formalism
for quantum preferences.
Elections are formulated as quantum circuits~\cite{NielsenC10}.
We define quantum constitutions
and five properties that constitutions can have.

\subsection{Hilbert-space formalism for quantum voters}

Let $\mathcal{S}$ denote a society of voters.
The voters are indexed by $i = 1, 2, \ldots N$.
We associate with voter $i$
the $i^\th$ copy of 
a Hilbert space $\mathcal{H}$. 
% (This association parallels the assignment of Hilbert spaces to prisoners
% in~\cite{EisertWL99}.)
The space of density operators 
(unit-trace linear positive-semidefinite operators) 
defined on $\mathcal{H}$
is denoted by $\Dens( \mathcal{H} )$. 

Society is associated with a \emph{joint quantum state} 
$\sigma_\soc  \in   
\Dens  \left( \mathcal{H}^{ \otimes N }  \right)$.
% Reference for notation: John Watrous's lecture notes -- Lec. 3 -- Search for "density."
% https://cs.uwaterloo.ca/~watrous/CS766/LectureNotes/03.pdf
The joint state encodes all the information
in the voters' preferences. 
This information may include correlations,
such as entanglement, between votes.
Consider tracing out every subsystem 
except the $i^\th$.
The result is voter $i$'s \emph{quantum preference},
$\rho_i  :=  \Tr_{\neq i} ( \sigma_\soc )$.
We sometimes denote a pure quantum preference by $\ket{ \rho_i }$. 
The set of all voters' quantum preferences 
forms society's \emph{quantum profile},
$\mathcal{P}  :=  \{ \rho_1, \ldots, \rho_N \}$.

Processing $\mathcal{P}$ must lead to
the identification of a winner.
More generally, the quantum society must generate
a transitive ordered list of the candidates.
We call such a list a \emph{classical preference}.
Each classical preference corresponds to
a state in $\mathcal{H}$.
For example, $c > a = b > d$ corresponds to 
$ \ket{ c{>}a{=}b{>}d }$. 
We denote by $\ket{\gamma}$ 
the $\gamma^{\rm th}$ classical-preference state 
and by $\chi_i^\gamma  :=  \ketbra{\gamma}{\gamma}$ 
the associated density operator.
The set $\{ \ket{ \gamma } \}$ forms 
the \emph{preference basis} $\BH$ 
for $\mathcal{H}$.

Consider any pair $(a, b)$ of candidates.
$\mathcal{H}$ decomposes into subspaces 
associated with the possible relationships 
between $a$ and $b$. 
By $\mathcal{G}^{a > b}$, 
we denote the subspace spanned by 
the $\BH$ elements 
that encode $a > b$
(e.g., $\ket{ a{>}b{=}c}$, $\ket{ c{>}a{>}b}$, etc.). % \footnote{
% < f >
% We will sometimes condense the subscript to
% $\mathcal{G}_{i }^{a > b}   :=
%  \mathcal{G}_{ \mathcal{H}_i }^{a > b}$
% and $\mathcal{G}_\soc^{a > b}   :=
%  \mathcal{G}_{ \mathcal{H}_\soc }^{a > b}$.
% These abbreviations will be used elsewhere, 
% as in notations for projectors and preference bases.}
% < /f >
The subspaces $\mathcal{G}^{b > a}$ and 
$\mathcal{G}^{a = b}$ are defined analogously. 
For example, $\ket{a{>}b{>}c}$ occupies 
the intersection of three subspaces:
$\ket{a{>}b{>}c} \in \mathcal{G}^{a > b}  \cap  
\mathcal{G}^{a > c}  \cap  
\mathcal{G}^{b > c}$.
The $a > b$, $b > a$, and $a = b$ subspaces are disjoint.
For example, $\mathcal{G}^{a > b}  
\cap  \mathcal{G}^{b > a}  =  \emptyset$.
$\Pi^{a > b}$ denotes 
the projector onto the subspace
$\mathcal{G}^{a > b}$.
The projector $\Pi^{a = b}$ 
is defined analogously.

Consider measuring projectively
a quantum preference $\rho_i$
with respect to $\BH$. 
The measurement yields a classical preference.
If $\rho_i$ is a nontrivial linear combination or mixture 
of $\BH$ elements, 
the measurement is probabilistic. 
A voter's ability to superpose classical preferences resembles 
a prisoner's ability to superpose classical tactics 
in the quantum Prisoner's Dilemma~\cite{EisertWL99}.
% Preference-basis measurements do not distinguish between quantum superpositions and classical mixtures. For example, the probability that a $\BH$ measurement of 
% \mbox{$\frac{1}{ \sqrt{2} } ( \ket{a > b}  +  \ket{b > a} )$} yields $a$ equals the probability that a $\BH$ measurement of
% \mbox{$\frac{1}{2} ( \ketbra{a > b}{a > b}  +  \ketbra{b > a}{b > a} )$} yields $a$, and the probability that one measurement yields $b$ equals the probability that the other measurement yields $b$.

During a \emph{quantum election}, 
society's joint state is transformed into 
\emph{society's quantum preference}:
$\sigma_\soc  \mapsto  \rho_\soc  \in  \Dens ( \Hil )$.
This $\rho_\soc$ is measured with respect to $\BH$.
generating society's classical preference.
The quantum election can be formulated
as a quantum circuit~\cite{NielsenC10}.
A quantum constitution, which we now introduce,
implements the transformation.

%
%
% Constitutions
%
%
\subsection{Quantum constitutions}
\label{section:Constitutions}

A \emph{classical constitution} $\C$ is a map from
% function $\mathcal{F} : V_1 \times \ldots \times V_N  \to  V_\soc$ that maps 
the profile of the voters' classical preferences
to society's classical preference. 
We define quantum constitutions analogously.
Having completed the definition of quantum elections,
we define their classical limit.
The classical constitutions that obey Arrow's Theorem
have four properties. 
We review and quantize these properties.

\subsubsection{Definition of ``quantum constitution''}

Quantum constitutions have the form of
general quantum evolutions,
as does the Quantum Prisoner's Dilemma~\cite{EisertWL99}.
A general quantum evolution is 
a convex-linear completely positive trace-preserving 
(CPTP) map~\cite{NielsenC10}.
A map $\mathcal{E}$ is convex-linear if,
given a probabilistic combination $\sum_i p_i \rho_i$ 
of states $\rho_i$,
$\mathcal{E}$ transforms the component states independently:
$\mathcal{E} \left(  \sum_i  p_i  \rho_i  \right)
   =  \sum_i  p_i  \mathcal{E} ( \rho_i ) \, ,$
wherein   $p_i  \geq  0  \;  \forall i$  and
$\sum_i  p_i  =  1$~\cite{NielsenC10}.
% Reference: Nielsen and Chuang, p. 367
% Reference: PSI academics -- Review courses -- Quantum Foundations Review by Rob Spekkens -- Powerpoints -- Lec. 3 -- p. 31
Every CPTP map is equivalent to
the tensoring on of an ancilla, 
a unitary transformation of the system-and-ancilla composite,
and the tracing out of a subsystem~\cite{NielsenC10}.

Each quantum constitution accepts, as input,
society's joint state, $\sigma_\soc$,
and an ancilla.
The ancilla is initialized to a fiducial state $\ketbra{0}{0}$.
When outputted by the constitution,
the ancilla holds society's quantum preference, $\rho_\soc$.

\begin{definition}[Quantum constitution]
\label{definition:Constitution}
A \emph{quantum constitution} is a convex-linear CPTP map 
\begin{align*}
   \mathcal{E} : 
   \Dens  \left( \mathcal{H}^{ \otimes (N + 1) }   \right)  
   \to   \Dens( \mathcal{H} )
\end{align*}
that transforms society's joint state and an ancilla 
into society's quantum preference:
\begin{equation}
   \mathcal{E}(\sigma_\soc  \otimes  \ketbra{0}{0} )  =  \rho_\soc  \, .
\end{equation}
\end{definition}

Having defined constitutions, we can define 
the classical limit.
\begin{definition} \label{definition:ClassLim}
The \emph{classical limit} of a quantum election is
the satisfaction of the following conditions:
\begin{enumerate}
   
   \item Every quantum voter preference $\rho_i$ is
   an element of the preference basis $\mathcal{B}$.
   
   \item The quantum constitution $\mathcal{E}$ consists of
   classical probabilistic logic gates.
   
\end{enumerate}
\end{definition}  \noindent
In the classical limit, $\mathcal{E}$ can output only
elements of $\mathcal{B}$ and probabilistic combinations thereof.

Classical and quantum constitutions can have various properties.
Four properties appear in Arrow's Theorem.
We review these classical properties, then quantize them.

%
%
% 4 properties
%
%
\subsubsection{The four constitutional properties in Arrow's Theorem
and quantum analogs}
\label{section:Properties}

Arrow's Theorem features four properties
available to classical constitutions:
transitivity, respecting of unanimity, 
respecting of independence of irrelevant alternatives, 
and being a dictatorship.
We review these properties and define quantum analogs.
% We review and quantize also a fifth property,
% respecting of majority rule.
% Quantum majority rule will provide a counterexample
% to the quantum analog of Arrow's Theorem.

Two principles guide the quantization strategy.
First, each quantum definition should preserve
the corresponding classical definition's spirit.
Second, each quantum definition should make sense 
in the context of entanglement and superpositions---should
be able to characterize a quantum circuit.

A classical constitution $\C$ is \emph{transitive} if 
every classical preference in its range is transitive. 
Suppose that society prefers candidate $a$ to $b$
and prefers $b$ to $c$.
$\C$ outputs a societal preference
in which $a$ ranks above $c$:
$a \geq b$ and $b \geq c$, together, imply $a \geq c$.
%
% Definition: Transitivity
%
\begin{definition}[Quantum transitivity]
\label{definition:Transitivity}
A quantum constitution $\mathcal{E}$ respects \emph{quantum transitivity} if every possible output $\rho_\soc$, upon being measured in 
the preference basis $\BH$, 
collapses to a state $\ket{a \ldots m}$ associated with 
a transitive classical preference $(a \ldots m)$.
\end{definition}
\noindent Every $\mathcal{E}$ obeys quantum transitivity by definition: 
Given any input, $\mathcal{E}$ outputs a $\rho_\soc$ that is 
a linear combination or a mixture of preference-basis elements. 
A $\BH$ measurement of $\rho_\soc$ yields 
a $\BH$ element.
Every $\BH$ element corresponds to 
a transitive classical preference.

Classical \emph{unanimity} is defined as follows.
Let $\C$ denote a classical constitution that respects unanimity.
Suppose that every voter ranks 
a candidate $a$ strictly above
a candidate $b$: $a > b$.
The constitution outputs a societal preference
in which $a$ ranks strictly above $b$: $a > b$.
%
% Definition: Unanimity
%
\begin{definition}[Quantum unanimity]
\label{definition:Unanimity}
A quantum constitution $\E$ respects \emph{quantum unanimity} if 
it has the following two subproperties:
\begin{enumerate}

   \item  \label{item:Have}
Suppose that every voter's quantum preference 
has support on the $a > b$ subspace:
$\Tr \left( \Pi^{a > b}  \:   \rho_i   \right)  >  0  \; \, 
\forall i  = 1, 2, \ldots, N$.
$\E$ outputs a societal quantum preference $\rho_\soc$ 
that has support on that subspace: 
$\Tr  \left(  \Pi^{a > b}  \:   \rho_\soc  \right)  >  0$.

   \item  \label{item:Only}
   Suppose that every voter's quantum preference 
has support only on the $a > b$ subspace.
$\E$ outputs a societal quantum preference $\rho_\soc$ 
that has support only on that subspace:
\begin{align}
   & \supp (\rho_i)  \subseteq  \Pi^{a > b}
   \quad \forall i = 1, 2, \ldots, N
   \quad  \Rightarrow  \quad
   \nonumber \\ & 
   \supp ( \rho_\soc )  \subseteq   \Pi^{a > b}  \, ,
\end{align}
wherein $\supp (\rho)$ denotes 
the support of the quantum state $\rho$.

\end{enumerate}
\end{definition}

Subproperty~\ref{item:Only} might appear extraneous,
seeming to lack a classical counterpart.
But classical unanimity satisfies
the classical analog of~\ref{item:Only} implicitly,
as the following argument shows.
\begin{enumerate}[label=(\Alph*)]
   \item
Suppose that every voter's preference 
satisfies the classical analog of
having support only on $\Pi^{a > b}$:
Every voter prefers $a > b$ strictly.
   \item  \label{item:Diff1}
Society prefers $a > b$ strictly,
by the definition of classical unanimity.
   \item  \label{item:Diff2}
Hence society ranks $a$ and $b$ 
neither as $a = b$ nor as $b > a$.
   \item
Hence society's preference 
satisfies the classical analog of
having support only on $ \Pi^{a > b}$.
\end{enumerate}
Definition~\ref{definition:Unanimity} must contain 
subproperty~\ref{item:Only} explicitly because
the quantum analog of step~\ref{item:Diff1}
does not imply
the quantum analog of step~\ref{item:Diff2}.
Even if $\rho_\soc$ has support on $\Pi^{a > b}$,
$\rho_\soc$ can have support on $\Pi^{a = b} \, :$
$\rho_\soc$ can be a linear combination of
elements of $B_\soc$
or can be a mixture.
The generality of quantum states
necessitates the articulation of subproperty~\ref{item:Only}.

Classical \emph{independence of irrelevant alternatives} (IIA)
is defined as follows.
In every classical preference, 
the candidates $a$ and $b$ have some \emph{relative ranking}.
Either $a > b$, $b > a$, or $a = b$.
Suppose that society's relative ranking of $a$ and $b$ 
depends only on
every voter's relative ranking of $a$ and $b$.
Whether society prefers $a$ to $b$ 
(or prefers $b$ to $a$, etc.)
depends only on
whether each voter prefers $a$ to $b$
(or prefers $b$ to $a$, etc.).
How any voter ranks candidate $c$
fails to influence society's relative ranking of $a$ and $b$.
%
% Definition: IIA
%
\begin{definition}[Quantum independence of irrelevant alternatives]
A quantum constitution respects 
\emph{quantum independence of irrelevant alternatives} (QIIA) if 
whether $\rho_\soc$ has support on 
$\mathcal{G}^{a > b}$, on 
$\mathcal{G}^{a < b}$, and/or on
$\mathcal{G}^{a = b}$ depends only on 
whether each $\rho_i$ has support on 
$\mathcal{G}^{a > b}$,  on $\mathcal{G}^{a < b}$,  and/or on 
$\mathcal{G}^{a = b}$.
\end{definition}

A classical dictatorship has a dominant voter. 
Suppose that society prefers $a$ strictly to $b$
if and only if
some voter $i$ prefers $a$ strictly to $b$,
for all pairs $(a, b)$ of candidates:
\begin{align}
   & \exists i  \:  :  \: 
   a > b \, ,  \;  \text{according to $i$,}
   \quad \Leftrightarrow \quad
   \nonumber \\ & \qquad  \;
   a > b \, ,  \;  \text{according to society,}  \quad
   \forall a, b  \, .
\end{align}
The classical constitution $\C$ that outputs society's preference
is a classical \emph{dictatorship}.
%
% Definition: Dictatorship
%
\begin{definition}[Quantum dictatorship]
\label{definition:Dictator}
A quantum constitution is a \emph{quantum dictatorship} if 
there exists a voter $i$ who has 
the following two characteristics: 
\begin{enumerate}
   \item
   Society's quantum preference has support on the $a > b$ subspace
   if and only if voter $i$'s has:
\begin{align}
   \Tr \left(   \Pi^{a > b}   \rho_i   \right)   >   0 
   \quad \Leftrightarrow \quad
   \Tr \left(   \Pi^{a > b}   \rho_\soc   \right)   >   0.
\end{align}

   \item  \label{item:Only2}
   Society's quantum preference has support
   only on the $a > b$ subspace
   if and only if voter $i$'s has:
\begin{align}
   \supp ( \rho_i )   \subseteq    \Pi^{a > b}
   \quad \Leftrightarrow \quad
   \supp ( \rho_\soc )   \subseteq   \Pi^{a > b}   \, .
\end{align}

\end{enumerate}
\end{definition}
\noindent Subproperty~\ref{item:Only2} plays a role analogous to 
subproperty~\ref{item:Only} in the definition of ``quantum unanimity.''

We have constructed quantum analogs of
the four properties in Arrow's Theorem.
A quantum version of Arrow's Theorem, we show,
is violated by
a quantum version of majority rule.

% The definition of ``quantum dictatorship'' lacks the ``if and only if'' in the definition of ``classical dictatorship.'' Even if $\rho_i$ has support on the $a > b$ subspace, $\rho_\soc$ can have support on the $a < b$ subspace.\footnote{
% < f >
% Earlier, we concluded that the lack of an ``if and only if'' follows from the definition of quantum constitution as outputting quantum states. But I believe that the ``quantum dictatorship'' definition can include an iff. Suppose that, if $\rho_i$ had support on the $a > b$ subspace, $\rho_\soc$ could not have support on the $a < b$ subspace. Suppose that $\ket{ \psi_i }  =  \frac{1}{ \sqrt{2} }  (  \ket{a{>}b}  +  \ket{b{>}a}  )$.
% The first term would prevent $\rho_\soc$ from having support on the $a < b$ subspace. The second term would prevent $\rho_\soc$ from having support on the $a > b$ subspace. $\rho_\soc$ could have support only on the $a = b$ subspace. Hence $\rho_\soc$ would be a quantum state. The constitution would be well-defined. Should we include the iff in the definition of ``quantum constitution,'' for increased consistency with the classical definition? Ning will probably reply ``don't bother.'' If you do reply so, why do you reply so?}
% < /f >

%
%
% Majority rule
%
%
\subsection{Majority rule}
\label{section:MajorityRule}

Majority rule is a fifth property that constitutions can have.
We review classical majority rule,
then introduce a quantum analog.
\emph{Cyclic} voting preferences
prevent classical majority rule 
from satisfying Arrow's assumptions.
Quantum majority rule is more robust.

\subsubsection{Classical majority rule}

Let $\mathcal{P}_\class$ denote a classical society's voter profile.
Let $\C$ denote a classical constitution
that respects majority rule.
$\C$ reflects the wishes shared by most voters.
Suppose that over half the voters agree
on the relative ranking of candidates $a$ and $b$.
$\C$ outputs a classical societal preference 
that has the same relative ranking of $a$ and $b$.

A subtlety arises if $\mathcal{P}_\class$ involves a cycle.
Let $T = \{a, b, \ldots, k\}$ denote
a set of candidates.
Suppose that $a$ and $b$ participate, in $\mathcal{P}_\class$,
in pairwise preferences that violate transitivity.
Suppose that every pair of candidates in $T$ does.
$T$ forms a classical \emph{cycle}.

For  example, let 
$\mathcal{P}_\class   =   \{ (a > b > c),  (c > a > b),  (b > c > a) \}.$
A na\"{i}ve application of majority rule implies 
$a > b$ and $b > c$.
Transitivity implies $a > c$. 
But a na\"{i}ve application of majority rule 
implies also $c > a$.
But $c > a$, combined with the previously derived $a > c$, 
violates transitivity. 
The constitution may be defined 
as outputting $a = b = c$ or 
as outputting an error message.\footnote{
% < f >
One profile can contain multiple cycles. For example, 
$\{ (a > b > c),  (b > a > c),  (a > c > b) \}$ contains a cycle over $(a, b)$ (because voters 1 and 3 rank $a > b$, whereas voter 2 ranks $b > a$) and a cycle over $(b, c)$ (because voters 1 and 2 rank $b > c$, whereas voter 3 ranks $c > b$).}
% < /f >

Cycles prevent classical majority rule from satisfying 
IIA and transitivity simultaneously.
Classical majority rule fails to satisfy Arrow's assumptions.
Hence classical majority rule
cannot contradict Arrow's Theorem.
A quantum analog of majority rule can.

\subsubsection{Quantum majority rule}

First, we introduce quantum cycles. 
We then define the Quantum Majority-Rule (QMR) constitution $\EQMR$.
This constitution, we show, respects
quantum transitivity, quantum unanimity, and QIIA.
These properties will enable $\EQMR$ to violate
a quantum analog of Arrow's Theorem.

\textbf{Quantum cycles:} 
Let $\chi_1^\alpha  \otimes  \ldots  \otimes  \chi_N^\mu$ be 
a product of preference-basis elements. 
Suppose that at least two 
$\chi_i^\gamma$'s are pure states labeled by  
classical preferences that form a classical cycle. 
The product will be said to contain a \emph{quantum cycle}.

\textbf{Operation of the Quantum Majority-Rule constitution:}
$\EQMR$ performs the following sequence of steps.
First, $\EQMR$ decoheres each quantum preference $\rho_i$
with respect to the preference basis:
\begin{align}
   \rho_i  \mapsto  \sum_{ \gamma }   
   \ketbra{\gamma}{\gamma}   \rho_i   \ketbra{\gamma}{\gamma}
   =   \sum_\gamma  p_i^\gamma  \chi_i^\gamma
   =:  \rho'_i \, ,
\end{align}
wherein  $\sum_\gamma  p_i^\gamma  =  1 \, .$
Society's quantum profile evolves as
\begin{align} 
   \label{eq:QMR_Help0a}
   \sigma_\soc  & \mapsto  
   \rho'_1  \otimes   \ldots   \otimes  \rho'_N   \\
   & =  \label{eq:QMR_Help1}
   \sum_{\alpha, \ldots, \mu}  
              \left(   p_1^\alpha   \ldots   p_N^\mu   \right)  
              (  \chi_1^\alpha   \otimes   \ldots   \otimes   \chi_N^\mu  )   \, .
\end{align}
Recall that $\chi_1^\alpha$ denotes
the element, labeled by the classical preference $\alpha$,
of the preference basis $\BH$ for 
voter 1's Hilbert space $\mathcal{H}$.

$\EQMR$, being a quantum constitution, obeys convex linearity.
To specify how $\EQMR$ transforms the right-hand side 
of Eq.~\eqref{eq:QMR_Help1},
we must specify just how $\EQMR$ transforms
each factor $\chi_1^\alpha   \otimes   \ldots   \otimes   \chi_N^\mu$.

For each factor, $\EQMR$ constructs a directed graph, or digraph.
One vertex is formed for each candidate.
The edges are governed by 
$\chi_1^\alpha   \otimes   \ldots   \otimes   \chi_N^\mu$.
If more classical preferences $\gamma$ correspond to $a > b$ 
than to $b > a$,
an edge points from $a$ to $b$.
If exactly as many $\gamma$'s correspond to $a > b$
as to $b > a$,
an edge points from $a$ to $b$ and from $b$ to $a$.

$\EQMR$ inputs the digraph into \emph{Tarjan's algorithm}~\cite{Tarjan72}.
Tarjan's algorithm finds a digraph's strongly connected components.
A \emph{strongly connected component} (SCC) is a subgraph.
Every vertex in the subgraph can be accessed from
every other vertex via edges.
Every vertex appears in exactly one SCC.
Every SCC in the QMR graph represents 
a cycle or a set of interlinked cycles.
For example, let $\chi_1^\alpha   \otimes   \ldots   \otimes   \chi_N^\mu
= \ketbra{b{>}a{>}c{>}d}{ b{>}a{>}c{>}d}  \otimes  
\ketbra{a{>}c{>}b{>}d }{ a{>}c{>}b{>}d }$.
Candidates $a$ and $b$ participate in a cycle, 
as do $b$ and $c$. 
The $a$, $b$, and $c$ vertices form an SCC.
The $d$ vertex forms another SCC.
The digraph appears in Fig.~\ref{fig:SCC}.

%
% Figure: Tarjan cycle
%
\begin{figure}[hbt]
\centering
\includegraphics[width=.5\textwidth, clip=true]{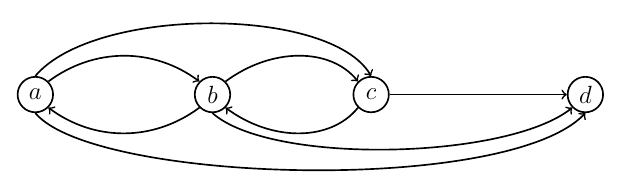}
\caption{\textbf{Example digraph formed by 
the Quantum Majority Rule (QMR) constitution $\EQMR$:}
Consider a society that consists of two voters.
Suppose that they submit the quantum preferences (votes)
$\ket{b{>}a{>}c{>}d}$ and $\ket{ a{>}c{>}b{>}d }$.
Voter 1 prefers candidate $b$ strictly to candidate $a$, etc.
Voter 2 prefers $a$ strictly to $c$, etc.
$\EQMR$ maps this set of preferences to a digraph.
Each candidate is associated with a vertex.
The voters' preferences determine the edges.
For example, most voters prefer candidate $a$ to candidate $d$.
Hence an edge points from vertex $a$ to vertex $d$.
Exactly as many voters prefer $a$ to $b$
as prefer $b$ to $a$.
Hence a doubly directed edge connects $a$ with $b$.
Vertices $a$ and $b$ form one cycle, 
while $b$ and $c$ form another.
$a$, $b$, and $c$ form 
a strongly connected component (SCC).
So does vertex $d$.
Tarjan's algorithm identifies digraphs' SCCs.
$\EQMR$ uses Tarjan's algorithm to form
society's quantum preference, 
to determine who wins the election.}
\label{fig:SCC}
\end{figure}

Tarjan's algorithm returns a list of the SCCs.
The later an SCC appears in the list,
the more popular the SCC's candidates, roughly speaking.
More precisely, let $i$ and $j$ label SCCs 
such that $i < j$.
Every vertex in the $j^{\rm th}$ SCC is preferred to 
every vertex in the $i^{\rm th}$.
For example, Tarjan's algorithm maps Fig.~\ref{fig:SCC} to 
\mbox{$( \{ d \},  \{a, b, c\} )$.}

Consider the strict classical preferences in which
every candidate in the $j^\th$ SCC ranks above
every candidate in the $i^\th$ SCC,
for all $j > i$.
$\EQMR$ forms a maximally mixed state $\rho_\soc'$
over the corresponding preference-basis elements.
In our example,
\begin{align}
   \label{eq:QMR_Help0}
   \rho_\soc'  & =  \frac{1}{ 6 }  
   (  \ketbra{ abcd }{ abcd }  +  \ketbra{ acbd }{ acbd }  +
      \ketbra{ bacd }{ bacd }  
      \nonumber \\ & \quad
      +  \ketbra{ bcad }{ bcad }
      +  \ketbra{ cabd }{ cabd }  +  \ketbra{ cbad }{ cbad }  )  \, .
\end{align}

$\EQMR$ then ``gives the minority a shot.''
For any candidate pair $(a, b)$,
suppose that at least one $\chi_i^\gamma$ corresponds 
to $a > b$.
The constitution spreads an amount $\delta \in (0, 1)$ of weight
across the $a  > b$ subspace:\footnote{
% < f >
``Giving the minority a shot'' resembles the action of 
the United States electoral college.
Whichever candidate receives the most popular votes
usually wins the presidential election.
But a candidate who receives a minority
can win the presidency
if favored by enough of the electoral college.}
% < /f >
\begin{align}
   \label{eq:Minority}
   \rho'_\soc  \mapsto  \rho''_\soc
   =  ( 1 - \delta )  \rho'_\soc  
   +  \delta   \,   \Pi^{a > b }  \, .
\end{align}
This $\delta$ serves as a parameter inputted to the constitution.
We omit $\delta$ from the notation $\EQMR$
for conciseness.

Next, $\EQMR$ enforces unanimity.
Suppose that every $\chi_i^\gamma$ corresponds to $a > b$,
for any candidate pair $(a, b)$:
$\supp ( \chi_i^\gamma )  \subseteq  \Pi^{a > b }  
\;  \;  \forall i  =  1, 2,  \ldots,  N$.
The constitution projects $\rho''_\soc$ onto 
the $a > b$ subspace:
\begin{align}
   \label{eq:Project_QMR}
   \rho''_\soc  \mapsto  \rho'''_\soc 
   =  \Pi^{a > b }    \:   \rho''_\soc   \:  \Pi^{a > b}  \, .
\end{align}

We have seen how $\EQMR$ calculates the $\rho'''_\soc$
associated with each term in Eq.~\eqref{eq:QMR_Help1}.
Each $ (  \chi_1^\alpha   \otimes   \ldots   \otimes   \chi_N^\mu  )$
in Eq.~\eqref{eq:QMR_Help1} is replaced with
the corresponding $\rho'''_\soc$.
This replacement yields $\rho_\soc$.
The $\rho_\soc$ is measured with respect to $\BH$.
The measurement yields
society's classical preference.

\textbf{Three properties of QMR:}
$\EQMR$ has three of the properties 
introduced in Sec.~\ref{section:Properties}.
These properties will enable $\EQMR$ to violate
a quantum analog of Arrow's Theorem.

\begin{lemma} \label{lemma:MajProperties}
The Quantum Majority-Rule constitution $\EQMR$ respects 
quantum transitivity, quantum unanimity, and 
quantum independence of irrelevant alternatives.
\end{lemma}

\begin{proof}
Every quantum constitution respects quantum transitivity,
as explained below Definition~\ref{definition:Transitivity}.
$\EQMR$ is a quantum constitution.
Therefore, $\EQMR$ respects quantum transitivity.

Quantum unanimity involves two subproperties 
(see Definition~\ref{definition:Unanimity}).
$\EQMR$ respects subproperty~\ref{item:Have}
due to Tarjan's algorithm and Eq.~\eqref{eq:QMR_Help0}.
Suppose that every voter's quantum profile
has support on the $a > b$ subspace.
Most quantum profiles have support on that subspace.
Hence $a$ appears in the $b$ SCC
or in an SCC ``preferred to'' the $b$ SCC.
Hence $\rho'_\soc$ contains preference-basis elements 
associated with $a > b$.

Equation~\eqref{eq:Project_QMR} ensures that
$\EQMR$ respects subproperty~\ref{item:Only} of quantum unanimity.
Suppose that every voter's quantum preference has support
only on the $a > b$ subspace.
Every $\chi_i^\gamma$ has support 
only on the $a > b$ subspace.
$\EQMR$ projects $\rho''_\soc$ onto $\mathcal{G}^{a > b }$,
not onto $\mathcal{G}^{b > a}$ or onto $\mathcal{G}^{a = b}$.
Therefore, $\rho'''_\soc$ has support only on $\mathcal{G}^{a > b }$.

According to QIIA, whether $\rho_\soc$ has support on
$\mathcal{G}^{a > c} \, $ on $\mathcal{G}^{c > a} \, $
and/or on $\mathcal{G}^{a = c}$
depends only on 
whether each voter's quantum preference, $\rho_i$,
has support on these subspaces---not
on whether any $\rho_i$ has support on, 
e.g., $\mathcal{G}^{a > b}$.
To check that $\EQMR$ respects QIIA,
we must analyze three cases:
\begin{enumerate}

   \item $a$ does not participate in a cycle with $c$.
   
   \begin{enumerate}
   
      \item \label{item:Case1a}
      $a$ participates in a cycle with 
      at least one candidate 
      that participates in a cycle with $c$.
      For example, $a$ may participate in a cycle with $b$,
      while $b$ participates in a cycle with $c$.
      
      \item \label{item:Case1b}
      $a$ participates in no cycle with any candidate
      that participates in a cycle with $c$.
      
    \end{enumerate}
    
    \item \label{item:Case2}
    $a$ participates in a cycle with $c$.

\end{enumerate}

Case~\ref{item:Case1a} requires the most thought.
Considering the example illustrated in Fig.~\ref{fig:SCC} suffices.
$a$ does not participate in a cycle with $c$.
Yet $a$ participates in a cycle with $b$,
which participates in a cycle with $c$.
Therefore, $a$ appears in the same SCC as $c$.
According to Eq.~\eqref{eq:QMR_Help0},
$\rho'_\soc$ has support on $\mathcal{G}^{c > a}$.
Yet every quantum voter preference $\rho_i$
has support only on $\mathcal{G}^{a > c}$.
How voters rank $b$ seems to influence
how society ranks $a$ relative to $c$.
$\EQMR$ seems to violate QIIA.

Equation~\eqref{eq:Project_QMR} rectifies
this seeming violation.
$\rho''_\soc$ is projected onto
the $a > c$ subspace,
because every $\chi_i^\gamma$ corresponds to $a > c$.

But suppose that not every $\chi_i^\gamma$ corresponded to $a > c$.
Suppose that only a majority of $\chi_i^\gamma$'s did.
$\rho''_\soc$ would not be projected onto $\mathcal{G}^{a > c}$.
How voters ranked $b$ would again seem to influence
how society ranked $a$ relative to $c$.
$\EQMR$ would again seem to violate QIIA.
$\EQMR$ would not because of Eq.~\eqref{eq:Minority}.
Some $\chi_i^\gamma$'s have support on 
the $c > a$ subspace.
The ``give the minority a shot'' step therefore
gives $\rho''_\soc$ support
on $\mathcal{G}^{c > a}$.
Society's quantum preference would have 
support on $\mathcal{G}^{c > a}$
regardless of whether 
$a$ participated in a cycle with
a $b$ that participated in a cycle with $c$.
How voters rank $b$ therefore does not affect
how society ranks $a$ relative to $c$.

In case~\ref{item:Case1b}, $a$ does not participate in
a cycle with any $b$ that participates in
a cycle with $c$.
Therefore, $a$ appears in an SCC
that ``is preferred'' to the $c$ SCC.
$\rho'_\soc$ therefore has support on 
just the $a > c$ subspace, 
regardless of any $b$'s.

In case~\ref{item:Case2}, $a$ participants in a cycle with $c$.
$\rho'_\soc$ has support on
the $a > c$ and $c > a$ subspaces,
regardless of any $b$'s.

\end{proof}
\noindent Because QMR satisfies the quantum analogs
of three properties in Arrow's Theorem, 
QMR can violate 
a quantum analog of Arrow's Theorem.

%
%
% Arrow's Theorem
%
%
\section{Arrow's Impossibility Theorem}
\label{section:ArrowTheorem}

Transitivity, unanimity, and IIA
have innocent-sounding definitions.
They seem unlikely to buttress authoritarianism.
Yet possessing these properties, Arrow shows,
renders a classical constitution a dictatorship~\cite{Arrow50}.
\begin{theorem}[Arrow's Impossibility Theorem]
Consider any (classical) constitution
used, with ranked voter preferences, 
to select from amongst at least three candidates.
If the constitution respects 
transitivity, unanimity, and independence of irrelevant alternatives,
the constitution is a dictatorship.
\end{theorem}
\noindent 
Multiple proof exist~\cite{Arrow50,Barbera_80_Pivotal,Geanakoplos05}.
Some involve a \emph{pivotal voter} $v$~\cite{Barbera_80_Pivotal,Geanakoplos05}.
If $v$ changes his/her mind
while all other preferences remain constant,
society's preference changes.
One proves first that the postulates imply
the existence of a voter slightly weaker than $v$.
This voter, one then shows, is pivotal and is a dictator.
No other dictator, one concludes, can exist.

We quantize Arrow's Theorem in the following conjecture.
\begin{conjecture}[Quantum Arrow Conjecture] 
\label{conjecture:Arrow}
Every quantum constitution that respects quantum transitivity, quantum unanimity, 
and quantum independence of irrelevant alternatives 
is a quantum dictatorship. 
\end{conjecture}
\begin{theorem} \label{theorem:Arrow}
The Quantum Arrow Conjecture is false.
\end{theorem}

\begin{proof}
We disprove the conjecture by counterexample.
The QMR constitution is combined with
a societal joint state $\sigma_\soc$
that encodes a cycle.
This combination, we show, lacks a dictator.
We have shown that QMR satisfies 
quantum transitivity, quantum unanimity, and QIIA.
Satisfying the conjecture's assumptions
but not its conclusion,
QMR and cyclic voting disprove the conjecture.

For simplicity, we focus on strict pairwise preferences.
We consider, e.g., $a > b$, ignoring $a = b$.
This focus frees us to drop binary-relation symbols:
\mbox{$\ket{a b c}  :=  \ket{a{>}b{>}c}$}.

%
% Figure: Disproof
%
\begin{figure}[hbt]
\centering
\includegraphics[width=.25\textwidth, clip=true]{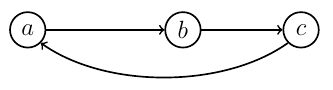}
\caption{\textbf{Digraph representation of the quantum votes
used to disprove the Quantum Arrow Conjecture:}
A quantum analog of majority rule,
acting on the quantum votes in Eq.~\eqref{eq:VoterPrefs},
violates a quantum analog of Arrow's Theorem. 
All three candidates---$a$, $b$, and $c$---form one cycle,
as depicted by the pattern of arrows.
Cycles prevent classical majority-rule constitutions
from satisfying all of Arrow's postulates.
Quantum majority rule is more robust.}
\label{fig:Disproof}
\end{figure}

Suppose that society's joint state is
a product that involves a cycle:
\begin{align}  \label{eq:VoterPrefs}
   \sigma_\soc  =  \ketbra{abc}{abc}  \otimes  \ketbra{ cab }{ cab }  
   \otimes  \ketbra{ bca }{ bca }  \, .
\end{align}
Decoherence relative to the preference basis
preserves the state.
$\EQMR$ constructs the digraph in Fig.~\ref{fig:Disproof}.
One edge points from $a$ to $b$ 
(because two voters prefer $a > b$, whereas one prefers $b > a$), 
one edge points from $b$ to $c$, 
and one edge points from $c$ to $a$. 
The digraph consists of one SCC.
$\EQMR$ therefore constructs the linear combination
\begin{align}
   \label{eq:Proof1}
   \rho'_\soc  & \propto  
   \ketbra{abc}{abc}  +  \ketbra{ cab }{ cab }  +  \ketbra{ bca }{ bca}  
   \nonumber \\ & \qquad
   +  \ketbra{ cba }{ cba }  +  \ketbra{ bac }{ bac }  +  \ketbra{ acb }{ acb }  \, .
\end{align}

The ``give the minority a shot'' step preserves the state:
$\rho''_\soc  =  \rho'_\soc \, .$
The voters do not unanimously prefer 
any candidate to any other:
For every voter $i$, there exists a voter $j$ such that
$\supp ( \rho_i )$ and $\supp ( \rho_j )$ occupy
subspaces labeled by distinct classical preferences.
$\EQMR$ therefore does not project $\rho''_\soc$
onto any subspace: $\rho'''_\soc  =  \rho''_\soc$.
Equation~\eqref{eq:QMR_Help1} consists of only one term,
so $\rho_\soc  =  \rho'''_\soc$.
Society's quantum preference appears in Eq.~\eqref{eq:Proof1}.

$\rho_\soc$ has support on multiple subspaces,
e.g., $\mathcal{G}^{a > b}$ and $\mathcal{G}^{b > a}$.
No quantum voter preference $\rho_i$ has.
No voter is a quantum dictator,
by definition~\ref{definition:Dictator}.
Yet $\mathcal{E}$ respects quantum transitivity, quantum unanimity, and QIIA, 
by Lemma~\ref{lemma:MajProperties}. 
$\EQMR$ satisfies the assumptions, but violates the conclusion,
of the Quantum Arrow Conjecture.
The conjecture is therefore false.
\end{proof}

%
% Analysis of counterexample
%

One can understand as follows why
our scheme violates the Quantum Arrow Conjecture.
The successes of quantum game theory motivate
the generalization of voting to accommodate
entangled and superposed preferences.
To introduce entanglement and superpositions,
one must formulate an election
as a general quantum process---a
preparation procedure, an evolution, and a measurement.
Classical constitutional properties must be quantized faithfully.
The quantum translations enable 
the Quantum Majority-Rule constitution to respect
quantum transitivity and QIIA simultaneously.
Classical majority-rule constitutions cannot respect
transitivity and IIA simultaneously,
due to cyclic votes.
But QMR satisfies all the assumptions 
in the Quantum Arrow Conjecture.
QMR, with a cyclic voter profile,
violates the conjecture.

Disproofs simpler than ours exist.
Ours offers interpretational advantages, however. 
For instance, let $\mathcal{K}$ denote a quantum constitution  
that outputs a superposition over all inputs.
$\mathcal{K}$ violates the conjecture. 
But imposing $\mathcal{K}$ on society---choosing 
society's classical preference totally randomly---makes 
little economic sense. 
Also, our disproof elucidates how quantization
invalidates Arrow's idea.
Classical majority rule fails to satisfy
Arrow's postulates, due to cycles.
Quantum Majority Rule is more resilient.
Quantization elevates
a classically inadequate disproof attempt
to a quantum disproof.

%
%
% Strategic voting
%
%
\section{Quantum voting tactics}
\label{section:Strategic}

% Reference: email from Ning on 9-24-14
% http://en.wikipedia.org/wiki/Tactical_voting
Imagine that Alice, Bob, and Charlie vie for 
the presidency of the American Physical Society. 
Suppose that Alice and Bob have 
greater chances of winning 
than Charlie has.
Suppose that Charlie agrees more with Alice than with Bob. 
Charlie's supporters might vote for Alice.
They would be trying to elect a president 
whom they neither prefer most nor mind most.
Charlie's supporters would be practicing strategic voting.
\emph{Strategic voting} is the submission of 
a preference other than one's opinion, 
to secure an unobjectionable outcome, 
in an election amongst three or more candidates~\cite{Barbera_01_Intro}.

We introduce \emph{quantum strategic voting}.
Voters leverage entanglement, superpositions, and interference.
% in elections (amongst three or more candidates).
We present three tactics reliant on entanglement
and one tactic reliant on interference and superpositions.
Other quantum tactics may exist and merit exploration.

To highlight the basic physics, 
we focus on strict preferences, 
as in the proof of Theorem~\ref{theorem:Arrow}.
For example, we consider $a > b$
to the exclusion of $a = b$.
We also focus on pure joint quantum states $\sigma_\soc$.

The strict-preferences assumption lets us compactify notation.
The classical preference  
$a > b > \ldots > k > \ell > m$
has the even permutations
$m > a > b > \ldots > k > \ell$,  \;
$\ell > m > a > b > \ldots > k$, etc.
These preferences are labeled $\alpha, \ldots, \mu \, .$ 
Each preference $\gamma$ corresponds to one anticycle.
We denote the anticycle with a bar: $\bar{\gamma}$.
For example, the cycle
$\alpha  :=  a > b > \ldots > m$
corresponds to the anticycle
$\bar{\alpha}  :=  m > \ldots  >  b  > a$.

Every pure quantum preference has the form
\begin{equation}   \label{eq:PureVote}
   \sum_{\gamma}  
   ( c_\gamma  \ket{ \gamma }   
   +   c_{ \bar{\gamma} }  \ket{ \bar{ \gamma } } ),
   \quad  {\rm  wherein}  \quad
   \sum_\gamma  \left( | c_\gamma |^2  
   +  | c_{ \bar{\gamma} } |^2   \right)  =  1.
\end{equation}
Society's joint quantum state has the form
\begin{align}
   \ket{ \sigma_\soc }  =
   (c_{\alpha_1}  \ldots  c_{\alpha_N} )   \ket{  \alpha  \ldots  \alpha }
   +  \ldots  +
   (c_{{ \bar{\mu} }_1}  \ldots  c_{{ \bar{\mu} }_N} )   \ket{  \bar{\mu}  \ldots  \bar{\mu} }.
\end{align}

%
%
% 3 entanglement uses
%
%
\subsection{Three entanglement-dependent voting tactics}

Let us simplify our quantum analog of majority rule,
now that QIIA, etc. need not concern us.
We introduce the variation \emph{QMR2}, labeled $\EQMRTwo$.
QMR2 is defined as follows.

$\EQMRTwo$ processes $\sigma_\soc$
as in Eqs.~\eqref{eq:QMR_Help0a} and~\eqref{eq:QMR_Help1}.
Society's joint quantum state $\sigma_\soc$ is decohered
with respect to the product of the voters' $\BH$'s.
$\EQMRTwo$ processes each term 
in Eq.~\eqref{eq:QMR_Help1} as follows.
The $j^\th$ term has the form 
\begin{align}
   \left(   p_1^\alpha   \ldots   p_N^\mu   \right)_j 
              (  \chi_1^\alpha   \otimes   \ldots   \otimes   \chi_N^\mu  )_j
   \equiv   p_j 
   (  \chi_1^\alpha   \otimes   \ldots   \otimes   \chi_N^\mu  )_j  \, .
\end{align}
The term is labeled by
a list $L_j  =  ( \alpha, \ldots, \mu )_j$ of classical preferences.
If most of the preferences are identical---if 
most equal $\gamma$, say---the 
$j^\th$ term in Eq.~\eqref{eq:QMR_Help1}
is associated with $( p_j ,  \gamma )$.
If no majority favors any $\gamma$, 
$\EQMRTwo$ chooses uniformly randomly from amongst 
the classical preferences that appear with 
the highest frequency in $L_j$.

$\EQMRTwo$ has assembled a list $( p_j,  \gamma_j )$.
Society's classical preference
is selected from amongst the $\gamma_j$'s
according to the probability distribution $\{ p_j \}$.

Entanglement can help one voter obstruct another. 
Imagine that the Supreme Court justices vote via QMR2. 
Suppose that Justice Alice wants to diminish Justice Bob's influence. 
However Bob votes, Alice should vote oppositely. 
Alice should entangle her quantum preference with Bob's.
(Given how opinionated Supreme Court justices are, Bob might not mind broadcasting his quantum preference.)
If Bob votes as in Eq.~\eqref{eq:PureVote}, Alice should form
\begin{align}   \label{eq:SupremeCourt}
   \sum_{\gamma}  
   ( c_\gamma  \ket{ \gamma  \,  \bar{\gamma} }   
   +   c_{ \bar{\gamma} }  \ket{ \bar{ \gamma }  \,  \gamma } )  \, .
\end{align}
Insofar as $\gamma$ represents Bob's preference, 
Alice votes oppositely, with $\bar{\gamma}$. 
Even if Bob changes his mind seconds before everyone votes, 
Alice need not scramble to alter her vote. 

Entanglement also facilitates party-line voting, if society uses QMR2. 
Suppose that Alice leads the Scientists' Party, 
to which Bob and Charlie belong. 
However Alice votes, Bob and Charlie wish to vote identically. 
The voters should form the entangled state
\begin{align}  \label{eq:GHZ}
   \sum_{\gamma}  
   ( c_\gamma  \ket{ \gamma  \,  \gamma  \,  \gamma }   
   +   c_{ \bar{\gamma} }  
   \ket{ \bar{ \gamma }    \, \bar{\gamma}   \,   \bar{\gamma} } )  \, ,
\end{align}
whose weights Alice chooses. 
This state generalizes the GHZ state: 
If the weights equal each other and only two candidates run,
\eqref{eq:GHZ} reduces to
$\frac{1}{ \sqrt{2} } ( \ket{ \alpha   \alpha  \alpha }  +  
  \ket{ \bar{\alpha}    \bar{\alpha}    \bar{\alpha}  } )$.
  
Finally, entangling voters' quantum preferences can 
pare down society's possible classical preferences. 
Suppose that Alice, Bob, and Charlie separately favor $\alpha$ 
twice as much as they prefer $\beta$. 
Each voter plans to submit 
$\sqrt{ \frac{2}{3} }  \, \ket{ \alpha }  +  \sqrt{ \frac{1}{3} }  \, \ket{ \beta }$. 
Society's joint state would be
\begin{align}  \label{eq:Entangle1}
   \ket{ \sigma_\soc }
   & =  \left(  \frac{2}{3} \right)^{ 3/2 }  \ket{ \alpha \alpha \alpha }
       +  \left(  \frac{1}{3} \right)^{ 3/2 }  \ket{ \beta   \beta   \beta }  
       \\ \nonumber &  \qquad
       +  \frac{2}{3^{3/2} }  (  \ket{  \beta  \alpha  \alpha  }   
       +   \ket{  \alpha  \beta  \alpha  }
               +  \ket{  \alpha  \alpha  \beta  }  )
       \\ \nonumber &  \qquad
       +  \frac{ \sqrt{2} }{ 3^{3/2} }  ( \ket{ \alpha  \beta  \beta }  
                   +  \ket{  \beta  \alpha  \beta  }   +   \ket{  \beta  \beta  \alpha  }  ).
\end{align}
If the constitution is QMR2, society might adopt $\alpha$ or $\beta$ 
as its classical preference. 

Suppose that Alice, Bob, and Charlie misunderstand entanglement. 
Eve can take advantage of their ignorance 
to eliminate $\beta$ from society's possible classical preferences. 
Suppose that Eve convinces the three citizens to submit
the W state
\begin{align}  \label{eq:Entangle2}
   \ket{ \sigma'_\soc }  =
   \frac{1}{ \sqrt{3} }  (  \ket{ \beta  \alpha  \alpha }   
      +   \ket{ \alpha  \beta  \alpha }   +   \ket{  \alpha  \alpha  \beta  }  ).
\end{align}
This entangled analog of $\ket{ \sigma_\soc }$, Eve might claim, 
represents the voters' opinion:
$\ket{ \sigma'_\soc }$ contains
twice as many $\alpha$'s as $\beta$'s. 
But QMR2 cannot map  $\ket{ \sigma'_\soc }$  to $\beta$. 
Entangled states lead to 
different possible election outcomes 
than product states.

%
%
% Interference + strategic voting
%
%
\subsection{Quantum strategic voting via interference}

Like entanglement, interference and relative phases 
facilitate quantum strategic voting. 
Consider a society $\mathcal{S}$ 
whose voters submit pure quantum preferences.
Let $\mathcal{S}$ use a variation \emph{QMR3},
denoted by $\EQMRTh$, on QMR. 

To illustrate QMR3 and the role of interference,
we consider the voter profile
\begin{align}
   \mathcal{P}  =  \left\{   \ket{abc},    \frac{1}{\sqrt{2}} (\ket{ bac } + \ket{ acb }),   
   \frac{1}{\sqrt{2}}  (\ket{ bac } + \ket{cba} )   \right\} \, .
\end{align}
Society's joint state has the form
\begin{align}
   \ket{ \sigma_\soc }  & =  
   \frac{1}{2} (  \ket{abc}  \ket{bac}  \ket{bac}   
   +   \ket{abc}  \ket{bac}  \ket{cba}
   \nonumber \\ & \quad
   +   \ket{abc}  \ket{acb}  \ket{bac}   
   +   \ket{abc}  \ket{acb}  \ket{cba}  )  \, .
\end{align}
Similarly to $\EQMR$, 
$\EQMRTh$ forms a digraph from 
each $\ket{ \sigma_\soc }$ term. 
Each digraph is inputted into Tarjan's algorithm, 
which returns a list of the SCCs. 
Just as $\EQMR$ maps each list to a mixed state $\rho'_\soc$, 
$\EQMRTh$ maps the $i^{\rm th}$ list 
to a superposition $\ket{ \rho_\soc^\I } \, $.
Society's quantum preference becomes
\begin{align}
   \ket{ \rho_\soc }  
   & \propto  \sum_{i = 1}^4  \ket{ \rho_\soc^\I }  \, .
\end{align}
In our example,
\begin{align}
   \ket{ \rho_\soc }  
   & =   \frac{1}{ \sqrt{6} } 
           ( \ket{bac}  +  \ket{bac}  +  \ket{abc}   +   \ket{acb} )  \\
   & =   \sqrt{ \frac{2}{3} }  \; \ket{bac}  
           +  \frac{1}{ \sqrt{6} } ( \ket{abc}   +   \ket{acb} )  \, .
\end{align}

$\ket{ \rho_\soc }$ may vanish:
The QMR3 quantum circuit 
may fail to output any quantum state.
If $\ket{ \rho_\soc }  =  0$, 
society can hold a revote. 
(Because QMR3 is defined on just pure states
and does not preserve all inputs' norms, 
QMR3 does not satisfy Definition~\ref{definition:Constitution}. 
QMR3 can be regarded as belonging to
an extension of quantum constitutions.)

Suppose that voter 3 wishes to eliminate $bac$ from 
society's possible classical preferences. 
Eliminating $\ket{ bac }$ from 
voter 3's quantum preference, $\ket{ \rho_3 }$, will not suffice. 
Voter 3 should introduce a relative phase of $-1$
into $\ket{ \rho_3 }$. 
(Alternatively, voter 3 could submit a superposition of 
$\ket{ abc }$ and $\ket{acb}$.)
Society's quantum profile becomes
\begin{align}
   \mathcal{P}'  =  \left\{   \ket{abc},    \frac{1}{\sqrt{2}} (\ket{ bac } + \ket{ acb }),   
   \frac{1}{\sqrt{2}}  (- \ket{ bac } + \ket{cba} )   \right\}.
\end{align}
Tarjan's algorithm leads to
$   \ket{ \rho_\soc }   
\propto   \frac{1}{2} ( - \ket{bac}  +  \ket{bac}  -  \ket{abc}   +   \ket{acb} ) \, .$
Hence $\ket{ \rho_\soc }  =  
\frac{1}{\sqrt{2}} ( \ket{abc}   -   \ket{acb})$.
Keeping the undesired $\ket{bac}$ in voter 3's quantum preference
contradicts our intuitions.
Yet interfering the new $\ket{ \rho_3 }$ with the other votes 
eliminates $bac$ from society's possible classical preferences.

%
%
% Conclusions
%
%
\section{Conclusions}

We have quantized elections, in the tradition of quantum game theory. 
The quantization obviates a quantum analog of Arrow's Theorem about the impossibility of a nondictatorship's having three simple properties. 
Entanglement, superpositions, and interference expand voters' arsenals of manipulation strategies. 
Whether other quantum strategies, unavailable to classical voters, exist 
merits investigation. 
So does whether monogamy of entanglement~\cite{CoffmanKW00} limits one voter's influence on others' quantum preferences. 
If creating entanglement is difficult (as in many labs), 
the resource theory of multipartite entanglement~\cite{Horodeckis09} 
might illuminate how voters can optimize their influence. 

Additionally, other voting schemes could be quantized.
Examples include proportional representation 
(in which the percentage of voters who favor Party $a$ dictates 
the number of government seats won by Party $a$) 
and cardinal voting (in which voters grade,
rather than rank, candidates).

Finally, counterstrategies may be formulated.
Consider our first entanglement-dependent voting example:
Justice Bob of the Supreme Court prepares his vote.
Justice Alice blocks Bob's effort using entanglement.
How should Justice Bob parry? Can entanglement assist him?
This problem mirrors quantum-cryptographic problems:
A sender wishes to communicate with a receiver securely.
An eavesdropper attacks.
The eavesdropper may access quantum or only classical resources,
depending on the problem.
How can the sender and receiver parry?
We have illustrated how our ``eavesdropper,'' Justice Alice,
might wield entanglement.
How Justice Bob should counter merits thought.

These opportunities can help illuminate
how quantum theory changes the landscape
of possible outcomes and strategies in games.

%
%
% Acknowledgements
%
%
\section*{Acknowledgements}
The authors thank Scott~Aaronson, Sarah~Harrison, Stephen~Jordan, Aleksander~Kubica, John~Preskill, and Thomas~Vidick for conversations. A Virginia Gilloon Fellowship, an IQIM Fellowship, a Burke Postdoctoral Fellowship, and NSF grant PHY-0803371. 
Further partial support came from 
the Walter Burke Institute for Theoretical Physics at Caltech.
The Institute for Quantum Information and Matter (IQIM) is an NSF Physics Frontiers Center supported by the Gordon and Betty Moore Foundation. NB thanks MIT, as well as QuICS at the U. of Maryland, for their hospitality during the completion of this work.
NYH thanks Fudan U. and National Taiwan U. similarly.

%
%
% Bibliography
%

\bibliographystyle{h-physrev}
\bibliography{VotingRefs}

% \begin{thebibliography}{99}
% \end{thebibliography}

\end{document}